\def\year{2021}\relax
\documentclass[letterpaper]{article} 
\usepackage{aaai21}  
\usepackage{times}  
\usepackage{helvet} 
\usepackage{courier}  
\usepackage[hyphens]{url}  
\usepackage{graphicx} 
\usepackage{natbib}  
\usepackage{caption} 
\nocopyright
\title{Dependency Stochastic Boolean Satisfiability:\\
\Large A Logical Formalism for NEXPTIME Decision Problems with Uncertainty\thanks{A condensed version of this work is published in the AAAI Conference on Artificial Intelligence (AAAI) 2021.}}

\author {
        Nian-Ze Lee,\textsuperscript{\rm 1}
        Jie-Hong R. Jiang\textsuperscript{\rm 1, 2}\\
}
\affiliations {
    \textsuperscript{\rm 1} Graduate Institute of Electronics Engineering,
    \textsuperscript{\rm 2} Department of Electrical Engineering\\
    National Taiwan University\\
    No. 1, Sec. 4, Roosevelt Rd., Taipei 10617, Taiwan\\
    \{d04943019, jhjiang\}@ntu.edu.tw
}

\usepackage{amssymb}
\usepackage{amsfonts}
\usepackage{amsmath}
\usepackage{amsthm}
\usepackage{adjustbox}
\usepackage{graphicx}
\usepackage{multirow}
\usepackage{multicol}
\usepackage{array}
\usepackage{tikz}
\usepackage{tikz-qtree}
\usepackage{tcolorbox}
\usepackage{cuted}
\usepackage{dsfont}
\usepackage{rotating}
\usepackage{booktabs}
\usepackage{subcaption}
\usepackage[pdftex,pdftitle={Dependency Stochastic Boolean Satisfiability: A Logical Formalism for NEXPTIME Decision Problems with Uncertainty},pdfauthor={Nian-Ze Lee, Jie-Hong R. Jiang}]{hyperref}
\hypersetup{bookmarksnumbered}
\usepackage[switch]{lineno}

\newcommand{\invR}{\begin{sideways}
\begin{sideways}$\mathsf{R}$\end{sideways}\end{sideways}}
\DeclareMathOperator*{\argmax}{arg\,max}

\newtheorem{theorem}{Theorem}

\newtheorem{corollary}{Corollary}

\newtheorem{example}{Example}

\begin{document}

\maketitle

\begin{abstract}
    \textit{Stochastic Boolean Satisfiability} (SSAT) is a logical formalism to model decision problems with uncertainty, such as \textit{Partially Observable Markov Decision Process} (POMDP) for verification of probabilistic systems.
    SSAT, however, is limited by its descriptive power within the PSPACE complexity class.
    More complex problems, such as the NEXPTIME-complete \textit{Decentralized POMDP} (Dec-POMDP), cannot be succinctly encoded with SSAT.
    To provide a logical formalism of such problems, we extend the \textit{Dependency Quantified Boolean Formula} (DQBF), a representative problem in the NEXPTIME-complete class, to its stochastic variant, named \textit{Dependency SSAT} (DSSAT), and show that DSSAT is also NEXPTIME-complete. We demonstrate the potential applications of DSSAT to circuit synthesis of probabilistic and approximate design.
    Furthermore, to study the descriptive power of DSSAT, we establish a polynomial-time reduction from Dec-POMDP to DSSAT.
    With the theoretical foundations paved in this work, we hope to encourage the development of DSSAT solvers for potential broad applications.
\end{abstract}

\section{Introduction}\label{sec:introduction}
Satisfiability (SAT) solvers~\cite{biere2009handbook} have been successfully applied to numerous research fields including artificial intelligence~\cite{nilsson2014principles,russell2016artificial}, electronic design automation~\cite{marques2000boolean,wang2009electronic}, software verification~\cite{berard2013systems,jhala2009software}, etc.
The tremendous benefits have encouraged the development of more advanced decision procedures for satisfiability with respect to more complex logics beyond pure propositional.
For example, solvers of the satisfiability modulo theories (SMT)~\cite{de2011satisfiability,barrett2018satisfiability} accommodate first order logic fragments; quantified Boolean formula (QBF)~\cite{qbflib,buning2009theory} allows both existential and universal quantifiers; stochastic Boolean satisfiabilty (SSAT)~\cite{littman2001stochastic,majercik2009stochastic} models uncertainty with random quantification; and dependency QBF (DQBF)~\cite{balabanov2014henkin,scholl2018dependency} equips Henkin quantifiers to describe multi-player games with partial information.
Due to their simplicity and generality, various satisfiability formulations are under active investigation.

Among the quantified decision procedures, QBF and SSAT are closely related.
While SSAT extends QBF to allow random quantifiers to model uncertainty, they are both PSPACE-complete~\cite{stockmeyer1973word}.
A number of SSAT solvers have been developed and applied in probabilistic planning, formal verification of probabilistic design, partially observable Markov decision process (POMDP), and analysis of software security.
For example, solver \texttt{MAXPLAN}~\cite{majercik1998maxplan} encodes a conformant planning problem as an exist-random quantified SSAT formula; solver \texttt{ZANDER}~\cite{majercik2003contingent} deals with partially observable probabilistic planning by formulating the problem as a general SSAT formula; solver \texttt{DC-SSAT}~\cite{majercik2005dc} relies on a divide-and-conquer approach to speedup the solving of a general SSAT formula.
Solvers \texttt{ressat} and \texttt{erssat}~\cite{lee2017solving,lee2018solving} are developed for random-exist and exist-random quantified SSAT formulas respectively, and applied to the formal verification of probabilistic design~\cite{lee2018towards}.
POMDP has also been studied under the formalism of SSAT~\cite{majercik2003contingent,SP19}.
Recently, bi-directional polynomial-time reductions between SSAT and POMDP are established~\cite{SP19}.
The quantitative information flow analysis for software security is also investigated as an exist-random quantified SSAT formula~\cite{fremont2017maximum}.

In view of the close relation between QBF and SSAT, we raise the question what would be the formalism that extends DQBF to the stochastic domain.
We formalize the \emph{dependency SSAT} (DSSAT) as the answer to the question.
We prove that DSSAT has the same NEXPTIME-complete complexity as DQBF~\cite{peterson2001lower}, and therefore it can succinctly encode decision problems with uncertainty in the NEXPTIME complexity class.

To highlight the benefits of DSSAT over DQBF, we note that DSSAT intrinsically represents an optimization problem (the answer is the maximum satisfying probability) while DQBF is a decision problem (the answer is either true or false).
The optimization nature of DSSAT potentially allows broader applications of the formalism.
Moreover, DSSAT is often preferable to DQBF in expressing problems involving uncertainty and probabilities.
As case studies, we investigate its applicability in probabilistic system design/verification and artificial intelligence.

In system design of the post Moore's law era, the practice of \emph{very large scale integration} (VLSI) circuit design experiences a paradigm shift in design principles to overcome the obstacle of physical scaling of computation capacity.
Probabilistic design~\cite{ProbDesign} and approximate design~\cite{ApproxDesign} are two such examples of emerging design methodologies.
The former does not require logic gates to be error-free, but rather allowing them to function with probabilistic errors.
The latter does not require the implementation circuit to behave exactly the same as the specification, but rather allowing their deviation to some extent.
These relaxations to design requirements provide freedom for circuit simplification and optimization.
We show that DSSAT can be a useful tool for the analysis of probabilistic design and approximate design.

The theory and applications of Markov decision process and its variants are among the most important topics in the study of artificial intelligence.
For example, the decision problem involving multiple agents with uncertainty and partial information is often considered as a decentralized POMDP (Dec-POMDP)~\cite{oliehoek2016concise}.
The independent actions and observations of the individual agents make POMDP for single-agent systems not applicable and require the more complex Dec-POMDP.
Essentially the complexity is lifted from the PSPACE-complete policy evaluation of finite-horizon POMDP to the NEXPTIME-complete Dec-POMDP.
We show that Dec-POMDP is polynomial time reducible to DSSAT.

To sum up, the main results of this work include:
\begin{itemize}
    \item formulating the DSSAT problem (Section~\ref{sec:DSSAT-def}),
    \item proving its NEXPTIME-completeness (Section~\ref{sec:DSSAT-complexity}), and
    \item showing its applications in:
          \begin{itemize}
              \item analyzing probabilistic/approximate design (Section~\ref{sec:partial})
              \item modeling Dec-POMDP (Section~\ref{sec:dec-pomdp}).
          \end{itemize}
\end{itemize}
Our results may encourage the development of DSSAT solvers to enable potential broad applications.
\section{Preliminaries}\label{sec:preliminaries}
In this section, we provide background knowledge about SSAT, DQBF, probabilistic design, and Dec-POMDP.

In the sequel, Boolean values \textsc{true} and \textsc{false} are represented by symbols $\top$ and $\bot$, respectively; they are also treated as $1$ and $0$, respectively, in arithmetic computation.
Boolean connectives $\neg, \vee, \wedge, \Rightarrow, \equiv$ are interpreted in their conventional semantics.
Given a set $V$ of variables, an \textit{assignment} $\alpha$ is a mapping from each variable $x\in V$ to $\mathbb{B}=\{\top,\bot\}$, and we denote the set of all assignments over $V$ by $\mathcal{A}(V)$. An assignment $\alpha$ \textit{satisfies} a Boolean formula $\phi$ over a set $V$ of variables if $\phi$ yields $\top$ after substituting all occurrences of every variable $x\in V$ with its assigned value $\alpha(x)$ and simplifying $\phi$ under the semantics of Boolean connectives. A Boolean formula $\phi$ over a set $V$ of variables is a \textit{tautology} if every assignment $\alpha\in \mathcal{A}(V)$ satisfies $\phi$.

\subsection{Stochastic Boolean Satisfiability}\label{subsec:ssat}
SSAT was first proposed by Papadimitriou and described as \textit{games against nature}~\cite{papadimitriou1985games}.
An SSAT formula $\Phi$ over a set $V=\{x_1,\ldots,x_n\}$ of variables is of the form: $Q_1 x_1, \ldots, Q_n x_n. \phi$, where each $Q_i \in \{\exists, \invR^{p}\}$ and Boolean formula $\phi$ over $V$ is quantifier-free.
Symbol $\exists$ denotes an existential quantifier, and $\invR^{p}$ denotes a randomized quantifier, which requires the probability that the quantified variable equals $\top$ to be $p\in[0,1]$.
Given an SSAT formula $\Phi$, the quantification structure $Q_1 x_1, \ldots, Q_n x_n$ is called the \emph{prefix}, and the quantifier-free Boolean formula $\phi$ is called the \emph{matrix}.

Let $x$ be the outermost variable in the prefix of an SSAT formula $\Phi$.
The satisfying probability of $\Phi$, denoted by $\Pr[\Phi]$, is defined recursively by the following four rules:
\begin{enumerate}
  \item[a)] $\Pr[\top]=1$,
  \item[b)] $\Pr[\bot]=0$,
  \item[c)] $\Pr[\Phi]=\max\{\Pr[\Phi|_{\neg x}], \Pr[\Phi|_{x}]\}$, if $x$ is existentially quantified,
  \item[d)] $\Pr[\Phi]=(1-p)\Pr[\Phi|_{\neg x}] + p\Pr[\Phi|_{x}]$, if $x$ is randomly quantified by $\invR^p$,
\end{enumerate}
where $\Phi|_{\neg x}$ and $\Phi|_{x}$ denote the SSAT formulas obtained by eliminating the outermost quantifier of $x$ via substituting the value of $x$ in the matrix with $\bot$ and $\top$, respectively.

The \textit{decision version} of SSAT is stated as follows. Given an SSAT formula $\Phi$ and a threshold $\theta \in [0,1]$, decide whether $\Pr[\Phi]\geq \theta$.
On the other hand, the \textit{optimization version} asks to compute $\Pr[\Phi]$.
The decision version of SSAT was shown to be PSPACE-complete~\cite{papadimitriou1985games}.

\subsection{Dependency Quantified Boolean Formula}\label{subsec:dqbf}
DQBF was formulated as \textit{multiple-person alternation}~\cite{peterson1979multiple}.
In contrast to the \textit{linearly ordered} prefix used in QBF, i.e., an existentially quantified variable will depend on all of its preceding universally quantified variables, the quantification structure in DQBF is extended with Henkin quantifiers, where the dependency of an existentially quantified variable on the universally quantified variables can be explicitly specified.

A DQBF $\Phi$ over a set $V=\{x_1,\ldots,x_n,y_1,\ldots,y_m\}$ of variables is of the form:
\begin{eqnarray}
  \forall x_1, \ldots, \forall x_n, \exists y_1(D_{y_1}), \ldots, \exists y_m(D_{y_m}). \phi,
\end{eqnarray}
where each $D_{y_j}\subseteq \{x_1,\ldots,x_n\}$ denotes the set of variables that variable $y_j$ can depend on, and Boolean formula $\phi$ over $V$ is quantifier-free.
We denote the set $\{x_1,\ldots,x_n\}$ (resp. $\{y_1,\ldots,y_m\}$) of universally (resp. existentially) quantified variables of $\Phi$ by $V_{\Phi}^\forall$ (resp. $V_{\Phi}^\exists$).

Given a DQBF $\Phi$, it is satisfied if for each variable $y_j$, there exists a function $f_j:\mathcal{A}(D_{y_j})\rightarrow \mathbb{B}$, such that after substituting variables in $V_{\Phi}^\exists$ with their corresponding functions respectively, matrix $\phi$ yields a tautology over $V_{\Phi}^\forall$.
The set of functions $\mathcal{F}=\{f_1,\ldots,f_m\}$ is called a set of \textit{Skolem} functions for $\Phi$.
In other words, $\Phi$ is satisfied by $\mathcal{F}$ if
\begin{linenomath}
  \begin{align}\label{eq:dqbf_min}
    \min_{\beta \in \mathcal{A}(V_{\Phi}^{\forall})}\mathds{1}_{\phi|_{\mathcal{F}}}(\beta)=1,
  \end{align}
\end{linenomath}
where $\mathds{1}_{\phi|_{\mathcal{F}}}(\cdot)$ is an indicator function to indicate whether an assignment over $V_{\Phi}^{\forall}$ belongs to the set of satisfying assignments of matrix $\phi$, when variables in $V_{\Phi}^{\exists}$ are substituted by their Skolem functions in $\mathcal{F}$.
That is, $\phi|_{\mathcal{F}} = \{\beta\ |\ \phi(\beta(x_1), \ldots, \beta(x_n)$, $f_1|_{\beta}$, $\ldots, f_m|_{\beta})\equiv \top\}$, where $f_j|_{\beta}$ is the logical value derived by substituting every $x_i\in D_{y_j}$ with $\beta(x_i)$ in function $f_j$.
The satisfiability problem of DQBF was shown to be NEXPTIME-complete \cite{peterson2001lower}.

\subsection{Probabilistic Design}
In this paper, a \textit{design} refers to a combinational Boolean logic circuit, which is a directed acyclic graph $G=(V,E)$, where $V$ is a set of vertices, and $E \subseteq V \times V$ is a set of edges.
Each vertex in $V$ can be a primary input, primary output, or an intermediate gate.
An intermediate gate is associated with a Boolean function.
An edge $(u,v) \in E$ signifies the connection from $u$ to $v$, denoting the associated Boolean function of $v$ may depend on $u$.
A circuit is called a \emph{partial design} if some of the intermediate gates are black boxes, that is, their associated Boolean functions are not specified.

\begin{figure}[t]
  \begin{center}
    \adjustbox{trim={.0\width} {.76\height} {.0\width} {.01\height},clip}{\includegraphics[width=0.5\textwidth]{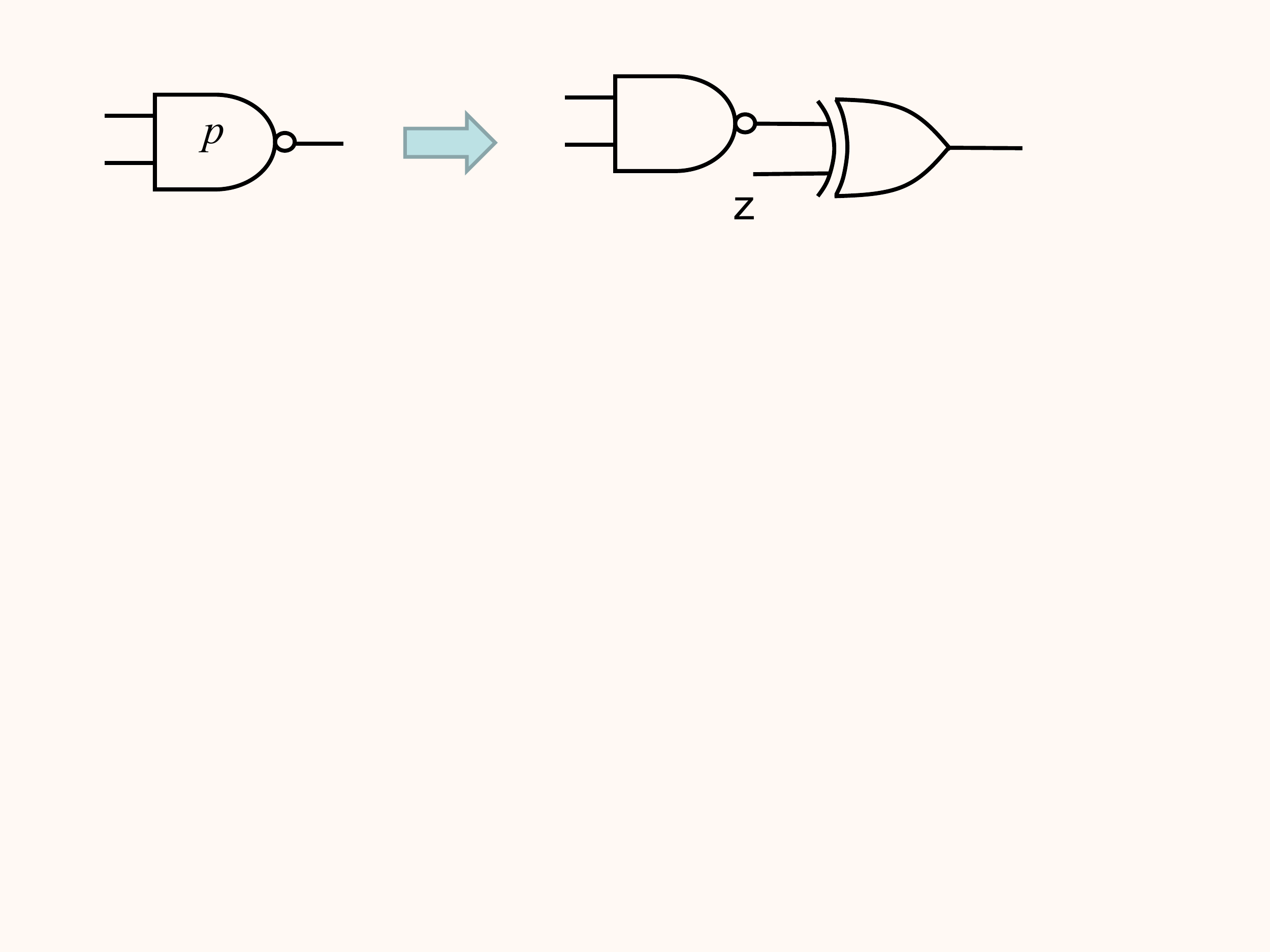}}
    \caption{Conversion of the distillation operation.}
    \label{fig:distill}
  \end{center}
  \vspace{-0.5cm}
\end{figure}

A \textit{probabilistic design} is an extension of conventional Boolean logic circuits to model the scenario where intermediate gates exhibit probabilistic behavior.
In a probabilistic design, each intermediate gate has an \textit{error rate}, i.e., the probability for the gate to produce an erroneous output.
An intermediate gate is \textit{erroneous} if its error rate is nonzero.
Using the \textit{distillation operation} \cite{lee2018towards}, an erroneous gate can be modeled by its corresponding error-free gate XORed with an auxiliary input, which valuates to $\top$ with a probability equal to the error rate.
As illustrated in Figure~\ref{fig:distill}, a NAND gate with error rate $p$ is converted to an error-free NAND gate XORed with a fresh auxiliary input $z$ with $\Pr[z = \top] = p$ so that it triggers the error with probability $p$.
After applying the distillation operation to every erroneous gate, all the intermediate gates in the distilled design become error-free, which makes the techniques for conventional Boolean circuit reasoning applicable.

\subsection{Decentralized POMDP}
Dec-POMDP is a formalism for multi-agent systems under uncertainty and with partial information.
Its computational complexity was shown to be NEXPTIME-complete~\cite{bernstein2002complexity}.
In the following, we briefly review the definition, optimality criteria, and value function of Dec-POMDP from the literature~\cite{oliehoek2016concise}.

A Dec-POMDP is specified by a tuple $\mathcal{M} = (I,S,\{A_i\},$ $T,  \rho, \{O_i\}, \Omega, \Delta_0, h)$, where
$I=\{1,\ldots,n\}$ is a finite set of $n$ agents,
$S$ is a finite set of states,
$A_i$ is a finite set of actions of Agent $i$,
$T: S \times (A_1 \times \cdots \times A_n) \times S \rightarrow [0,1]$ is a transition distribution function with
$T(s,\Vec{a},s')=\Pr[s'|s,\vec{a}]$, the probability to transit to state $s'$ from state $s$ after taking actions $\vec{a}$,
$\rho: S \times (A_1 \times \cdots \times A_n) \rightarrow \mathbb{R}$ is a reward function with $\rho(s, \vec{a})$ giving the reward for being in state $s$ and taking actions $\vec{a}$,
$O_i$ is a finite set of observations for Agent $i$,
$\Omega: S \times (A_1 \times \cdots \times A_n) \times (O_1 \times \cdots \times O_n) \rightarrow [0,1]$ is an observation distribution function with
$\Omega(s',\Vec{a},\vec{o})=\Pr[\vec{o}|s',\vec{a}]$, the probability to receive observation $\vec{o}$ after taking actions $\vec{a}$ and transiting to state $s'$, $\Delta_0: S \rightarrow [0,1]$ is an initial state distribution function with $\Delta_0(s)=\Pr[s^0 \equiv s]$, the probability for the initial state $s^0$ being state $s$, and $h$ is a planning horizon, which we assume finite in this work.

Given a Dec-POMDP $\mathcal{M}$, we aim at maximizing the expected total reward $\mathbb{E}[\sum_{t=0}^{h-1}\rho(s^t,\vec{a}^t)]$ through searching an optimal \textit{joint policy} for the team of agents.
Specifically, a \textit{policy} $\pi_i$ of Agent $i$ is a mapping from the agent's \textit{observation history}, i.e., a sequence of observations $\underline{o_i^t}=o_i^0,\ldots,o_i^t$ received by Agent $i$, to an action $a_i^{t+1}\in A_i$.
A joint policy for the team of agents $\vec{\pi}=(\pi_1,\ldots,\pi_n)$ maps the agents' joint observation history $\vec{\underline{o}}^t=(\underline{o_1^t},\ldots,\underline{o_n^t})$ to actions $\vec{a}^{t+1}=(\pi_1(\underline{o_1^t}),\ldots,\pi_n(\underline{o_n^t}))$.
We shall focus on deterministic policies only, as it was shown that every Dec-POMDP with a finite planning horizon has a deterministic optimal joint policy~\cite{oliehoek2008optimal}.

To assess the quality of a joint policy $\vec{\pi}$, its \textit{value} is defined to be $\mathbb{E}[\sum_{t=0}^{h-1}\rho(s^t,\vec{a}^t)|\Delta_0,\vec{\pi}]$.
The \textit{value function} $V(\vec{\pi})$ can be computed in a recursive manner, where for $t=h-1$, $V^\pi(s^{h-1},\vec{\underline{o}}^{h-2})=\rho(s^{h-1},\vec{\pi}(\vec{\underline{o}}^{h-2}))$, and for $t<h-1$,
\begin{linenomath}
  \begin{align}
     & V^\pi(s^t,\vec{\underline{o}}^{t-1})=\rho(s^t,\vec{\pi}(\vec{\underline{o}}^{t-1}))+\nonumber                                                                     \\
     & \sum_{s^{t+1}\in S}\sum_{\vec{o}^t\in\vec{O}}\Pr[s^{t+1},\vec{o}^t|s^t,\vec{\pi}(\vec{\underline{o}}^t)]V^\pi(s^{t+1},\vec{\underline{o}}^{t}).\label{eq:bellman}
  \end{align}
\end{linenomath}
The probability $\Pr[s^{t+1},\vec{o}^{t}|s^t,\vec{\pi}(\vec{\underline{o}}^t)]$ in the above equation is the product of $T(s^t,\vec{\pi}(\vec{\underline{o}}^t),s^{t+1})$ and $\Omega(s^{t+1},\vec{\pi}(\vec{\underline{o}}^t),\vec{o}^{t})$.
Eq.~\eqref{eq:bellman} is also called the \textit{Bellman Equation} of Dec-POMDP.

Denoting the empty observation history at the first stage (i.e., $t=0$) with the symbol $\vec{\underline{o}}^{-1}$, the value $V(\vec{\pi})$ of a joint policy equals $\sum_{s^0\in S}\Delta_0(s^0)V^\pi(s^0,\vec{\underline{o}}^{-1})$.

\section{DSSAT Formulation}\label{sec:DSSAT-def}
In this section, we extend DQBF to its stochastic variant, named \textit{Dependency Stochastic Boolean Satisfiability} (DSSAT).

A DSSAT formula $\Phi$ over $V=\{x_1,\ldots,x_n,y_1,\ldots,y_m\}$ is of the form:
\begin{align} \label{eq:dssat}
    \invR^{p_1} x_1, \ldots, \invR^{p_n} x_n, \exists y_1(D_{y_1}), \ldots, \exists y_m(D_{y_m}). \phi,
\end{align}
where each $D_{y_j}\subseteq \{x_1,\ldots,x_n\}$ denotes the set of variables that variable $y_j$ can depend on, and Boolean formula $\phi$ over $V$ is quantifier-free.
We denote the set $\{x_1,\ldots,x_n\}$ (resp. $\{y_1,\ldots,y_m\}$) of randomly (resp. existentially) quantified variables of $\Phi$ by $V_{\Phi}^{\invR}$ (resp. $V_{\Phi}^\exists$).

Given a DSSAT formula $\Phi$ and a set of Skolem functions $\mathcal{F}=\{f_j:\mathcal{A}(D_{y_j})\rightarrow \mathbb{B}\ |\ j=1,\ldots,m\}$, the satisfying probability $\Pr[\Phi|_{\mathcal{F}}]$ of $\Phi$ with respect to $\mathcal{F}$ is defined by the following equation:
\begin{linenomath}
    \begin{align}\label{eq:dssat_sum}
        \Pr[\Phi|_{\mathcal{F}}]=\sum_{\alpha \in \mathcal{A}(V_{\Phi}^{\invR})}\mathds{1}_{\phi|_{\mathcal{F}}}(\alpha)w(\alpha),
    \end{align}
\end{linenomath}
where $\mathds{1}_{\phi|_{\mathcal{F}}}(\cdot)$ is the indicator function defined in Section~\ref{subsec:dqbf} and $w(\alpha)=\prod_{i=1}^n p_i^{\alpha(x_i)}(1-p_i)^{1-\alpha(x_i)}$ is the weighting function for assignments.
In other words, the satisfying probability is the summation of weights of satisfying assignments over $V_{\Phi}^{\invR}$.
The weight of an assignment can be understood as its occurring probability in the space of $\mathcal{A}(V_{\Phi}^{\invR})$.

The \textit{decision version} of DSSAT is stated as follows.
Given a DSSAT formula $\Phi$ and a threshold $\theta \in [0,1]$, decide whether there exists a set of Skolem functions $\mathcal{F}$ such that $\Pr[\Phi|_{\mathcal{F}}]\geq \theta$.
On the other hand, the \textit{optimization version} asks to find a set of Skolem functions to maximize the satisfying probability of $\Phi$.

The formulation of SSAT can be extended by incorporating universal quantifiers, resulting in a unified framework named \textit{extended} SSAT 
\cite{majercik2009stochastic}, which subsumes both QBF and SSAT.
In the extended SSAT, besides the four rules discussed in Section~\ref{subsec:ssat} for calculating the satisfying probability of an SSAT formula $\Phi$, the following rule is added: $\Pr[\Phi]=\min\{\Pr[\Phi|_{\neg x}], \Pr[\Phi|_{x}]\}$, if $x$ is universally quantified.
Similarly, an \textit{extended} DSSAT 
formula $\Phi$ over a set of variables
$\{x_1,\ldots,x_n,y_1,\ldots,y_m,z_1,\ldots,z_l\}$ is of the form:
\begin{linenomath}
    \begin{align}\label{eq:xdssat}
        Q_1 v_1,\ldots, Q_{n+l} v_{n+l},\exists y_1(D_{y_1}),\ldots, \exists y_m(D_{y_m}). \phi,
    \end{align}
\end{linenomath}
where $Q_i v_i$ equals either $\invR^{p_k} x_k$ or $\forall z_k$ for some $k$ with $v_i \neq v_j$ for $i \neq j$,
and each $D_{y_j}\subseteq \{x_1,\ldots,x_n,z_1,\ldots,z_l\}$ denotes the set of randomly and universally quantified variables which variable $y_j$ can depend on.
The satisfying probability of $\Phi$ with respect to a set of Skolem functions $\mathcal{F}=\{f_j:\mathcal{A}(D_{y_j})\rightarrow \mathbb{B}\ |\ j=1,\ldots,m\}$, denoted by $\Pr[\Phi|_{\mathcal{F}}]$, can be computed by recursively applying the aforementioned five rules to the induced formula of $\Phi$ with the existential variables $y_j$ being substituted with their respective Skolem functions $f_j$.
Under the above computation scheme, both Eq.~\eqref{eq:dqbf_min} and Eq.~\eqref{eq:dssat_sum} are special cases, where the variables preceding the existential quantifiers in the prefixes are solely universally or randomly quantified, and hence the fifth or the fourth rule is applied to calculate $\Pr[\Phi|_{\mathcal{F}}]$.

Note that in the above extension the Henkin-type quantifiers are only defined for the existential variables.
Although the extended formulation increases practical expressive succinctness, the computational complexity is not changed as to be shown in the next section.

\section{DSSAT Complexity}\label{sec:DSSAT-complexity}
In the following, we show that the decision version of DSSAT is NEXPTIME-complete.
\begin{theorem}
    DSSAT is NEXPTIME-complete.
\end{theorem}
\begin{proof}
    To show that DSSAT is NEXPTIME-complete, we have to show that it belongs to the NEXPTIME complexity class and that it is NEXPTIME-hard.

    First, to see why DSSAT belongs to the NEXPTIME complexity class, observe that a Skolem function for an existentially quantified variable can be guessed and constructed in nondeterministic exponential time with respect to the number of randomly quantified variables.
    Given the guessed Skolem functions, the evaluation of the matrix, summation of weights of satisfying assignments, and comparison against the threshold $\theta$ can also be performed in exponential time.
    Overall, the whole procedure is done in nondeterministic exponential time with respect to the input size, and hence DSSAT belongs to the NEXPTIME complexity class.

    Second, to see why DSSAT is NEXPTIME-hard, we reduce the NEXPTIME-complete problem DQBF to DSSAT as follows. Given an arbitrary DQBF:
    \[
        \Phi_Q=\forall x_1, \ldots, \forall x_n, \exists y_1(D_{y_1}), \ldots, \exists y_m(D_{y_m}).\phi,
    \]
    we construct a DSSAT formula:
    \[
        \Phi_S=\invR^{0.5} x_1, \ldots, \invR^{0.5} x_n, \exists y_1(D_{y_1}), \ldots, \exists y_m(D_{y_m}).\phi
    \]
    by changing every universal quantifier to a randomized quantifier with probability $0.5$.
    The reduction can be done in polynomial time with respect to the size of $\Phi_Q$.
    We will show that $\Phi_Q$ is satisfiable if and only if there exists a set of Skolem functions $\mathcal{F}$ such that $\Pr[\Phi_S|_{\mathcal{F}}]\geq 1$.

    The ``only if'' direction: As $\Phi_Q$ is satisfiable, there exists a set of Skolem functions $\mathcal{F}$ such that after substituting the existentially quantified variables with the corresponding Skolem functions, matrix $\phi$ becomes a tautology over variables $\{x_1,\ldots,x_n\}$.
    Therefore, $\Pr[\Phi_S|_{\mathcal{F}}]=1\geq 1$.

    The ``if'' direction: As there exists a set of Skolem functions $\mathcal{F}$ such that $\Pr[\Phi_S|_{\mathcal{F}}]\geq 1$, after substituting the existentially quantified variables with the corresponding Skolem functions, each assignment $\alpha\in \mathcal{A}(\{x_1,\ldots,x_n\})$ must satisfy $\phi$, i.e., $\phi$ becomes a tautology over variables $\{x_1,\ldots,x_n\}$.
    Otherwise, the satisfying probability $\Pr[\Phi_S|_{\mathcal{F}}]$ must be less than $1$ as the weight $2^{-n}$ of some unsatisfying assignment is missing from the summation.
    Therefore, $\Phi_Q$ is satisfiable.
\end{proof}

When DSSAT is extended with universal quantifiers,
its complexity remains in the NEXPTIME complexity class as the fifth rule of the satisfying probability calculation does not incur any complexity overhead.
Therefore the following corollary is immediate.
\begin{corollary}
    The decision problem of DSSAT extended with universal quantifiers of Eq.~\eqref{eq:xdssat} is NEXPTIME-complete.
\end{corollary}

\section{Application: Analyzing Probabilistic and Approximate Partial Design}\label{sec:partial}
After formulating DSSAT and proving its NEXPTIME-completeness, we show its application to the analysis of probabilistic design and approximate design.
Specifically, we consider the probabilistic version of the \emph{topologically constrained logic synthesis problem}~\cite{Sinha:2002,balabanov2014henkin}, or equivalently the \emph{partial design problem}~\cite{ICCD13}.

In the \emph{(deterministic) partial design problem}, we are given a specification function $G(X)$ over primary input variables $X$ and a \textit{partial design} $C_F$ with black boxes to be synthesized.
The Boolean functions induced at the primary outputs of $C_F$ can be described by $F(X,T)$, where $T$ corresponds to the variables of the black box outputs.
Each black box output $t_i$ is specified with its input variables (i.e., dependency set) $D_i \subseteq X \cup Y$ in $C_F$, where $Y$ represents the variables for intermediate gates in $C_F$ referred to by the black boxes.
The partial design problem aims at deriving the black box functions $\{h_1(D_1), \ldots, h_{|T|}(D_{|T|})\}$ such that substituting $t_i$ with $h_i$ in $C_F$ makes the resultant circuit function equal $G(X)$.
The above partial design problem can be encoded as a DQBF problem; moreover, the partial equivalence checking problem is shown to be NEXPTIME-complete~\cite{ICCD13}.

Specifically, the DQBF that encodes the partial equivalence checking problem is of the form:
\begin{align}\label{eq:dqbf-pd}
   & \forall X,\forall Y,\exists T(D).\nonumber       \\
   & (Y \equiv E(X)) \rightarrow (F(X,T)\equiv G(X)),
\end{align}
where $D$ consists of $(D_1, \ldots, D_{|T|})$, $E$ corresponds to the defining functions of $Y$ in $C_F$, and the operator ``$\equiv$'' denotes elementwise equivalence between its two operands.

\begin{figure}[t]
  \begin{center}
    \adjustbox{trim={.08\width} {.59\height} {.54\width} {.01\height},clip}{\includegraphics[width=0.95\textwidth]{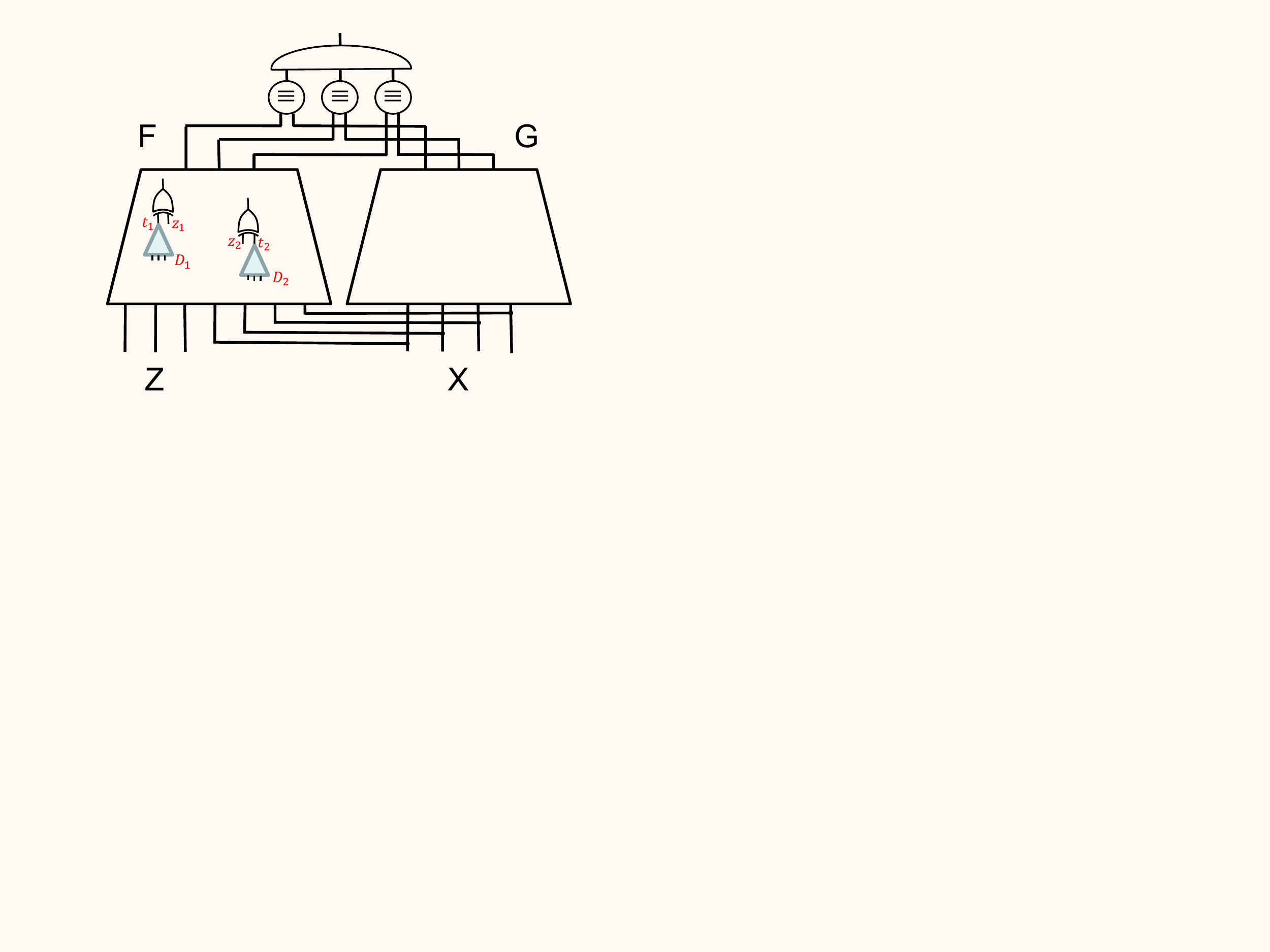}}
    \caption{Circuit for the equivalence checking of probabilistic partial design.}
    \label{fig:miter}
  \end{center}
  \vspace{-0.5cm}
\end{figure}

The above partial design problem can be extended to its probabilistic variant, which is illustrated by the circuit shown in Figure~\ref{fig:miter}.
The \emph{probabilistic partial design problem} is the same as the deterministic partial design problem except that
$C_F$ is a distilled probabilistic design \cite{lee2018towards} with black boxes, whose functions at the primary outputs can be described by $F(X,Z,T)$, where $Z$ represents the variables for the auxiliary inputs that trigger errors in $C_F$ (including the errors of the black boxes) and $T$ corresponds to the variables of the black box outputs.
Each black box output $t_i$ is specified with its input variables (i.e., dependency set) $D_i \subseteq X \cup Y$ in $C_F$.
When $t_i$ is substituted with $h_i$ in $C_F$, the function of the resultant circuit is required to be sufficiently close to the specification with respect to some expected probability.

\begin{theorem}\label{thm:partial}
  The probabilistic partial design equivalence checking problem is NEXPTIME-complete.
\end{theorem}
\begin{proof}
  To show that the probabilistic partial design problem is in the NEXPTIME complexity class,
  we note that the black box functions can be guessed and validated in time exponential to the number of black box inputs.

  To show completeness in the NEXPTIME complexity class,
  we reduce the known NEXPTIME-complete DSSAT problem to the probabilistic partial design problem,
  similar to the construction used in the previous work~\cite{ICCD13}.
  Given a DSSAT instance, it can be reduced to a probabilistic partial design instance in polynomial time as follows.
  Without loss of generality, consider the DSSAT formula of Eq.~\eqref{eq:dssat}.
  We create a probabilistic partial design instance by letting the specification $G$ be a tautology and letting $C_F$ be a probabilistic design with black boxes,
  which involves primary inputs $x_1, \ldots, x_n$ and black box outputs $y_1, \ldots, y_m$ to compute the matrix $\phi$.
  The driving inputs of the black box output $y_j$ is specified by the dependency set $D_{y_j}$ in Eq.~\eqref{eq:dssat}, and the probability for primary input $x_i$ to evaluate to $\top$ is set to $p_i$.
  The original DSSAT formula is satisfiable with respect to a target satisfying probability $\theta$ if and only if there exist implementations of the black boxes such that the resultant circuit composed with the black box implementations behaves like a tautology with respect to the required expectation $\theta$.
\end{proof}

On the other hand, the probabilistic partial design problem can be encoded with the following DSSAT formula
\begin{linenomath}
  \begin{align}\label{eq:dssat-pd}
     & \invR X,\invR Z,\forall Y,\exists T(D).\nonumber   \\
     & (Y \equiv E(X)) \rightarrow (F(X,Z,T)\equiv G(X)),
  \end{align}
\end{linenomath}
where the primary input variables are randomly quantified with probability $p_{x_i}$ of $x_i \in X$ to reflect their weights, and the error-triggering auxiliary input variables $Z$ are randomly quantified according to the pre-specified error rates of the erroneous gates in $C_F$. Notice that the above DSSAT formula takes advantage of the extension with universal quantifiers as discussed previously.


In approximate design, a circuit implementation may deviate from its specification by a certain extent.
The amount of deviation can be characterized in a way similar to the error probability calculation in probabilistic design.
For approximate partial design, the equivalence checking problem can be expressed by the DSSAT formula:
\begin{linenomath}
  \begin{align}\label{eq:dssat-ad}
     & \invR X,\forall Y,\exists T(D).\nonumber         \\
     & (Y \equiv E(X)) \rightarrow (F(X,T)\equiv G(X)),
  \end{align}
\end{linenomath}
which differs from Eq.~\eqref{eq:dssat-pd} only in requiring no auxiliary inputs.
The probabilities of the randomly quantified primary input variables are determined by the approximation criteria in measuring the deviation.
For example, when all the input assignments are of equal weight, the probabilities of the primary inputs are all set to 0.5.

We note that as the engineering change order (ECO) problem \cite{DATE20} heavily relies on partial design equivalence checking, the above DSSAT formulations provide fundamental bases for ECOs of probabilistic and approximate designs.
\section{Application: Modeling Dec-POMDP}\label{sec:dec-pomdp}
\begin{figure*}[t]
    \centering
    \input{fig/formulas}
    \caption{The formulas used to encode a Dec-POMDP $\mathcal{M}$.}\label{fig:formula}
    \vspace{-0.5cm}
\end{figure*}

In this section we demonstrate the descriptive power of DSSAT to model NEXPTIME-complete problems by constructing a polynomial-time reduction from Dec-POMDP to DSSAT.
Our reduction is an extension of that from POMDP to SSAT proposed in the previous work~\cite{SP19}.

In essence, given a Dec-POMDP $\mathcal{M}$, we will construct in polynomial time a DSSAT formula $\Phi$ such that there is a joint policy $\vec{\pi}$ for $\mathcal{M}$ with value $V(\vec{\pi})$ if and only if there is a set of Skolem functions $\mathcal{F}$ for $\Phi$ with satisfying probability $\Pr[\Phi|_{\mathcal{F}}]$, and $V(\vec{\pi})=\Pr[\Phi|_{\mathcal{F}}]$.

First we introduce the variables used in construction of the DSSAT formula and their domains.
To improve readability, we allow a variable $x$ to take values from a finite set $U=\{x_1,\ldots,x_K\}$~\cite{SP19}.
Under this setting, a randomized quantifier $\invR$ over variable $x$ specifies a distribution $\Pr[x\equiv x_i]$ for each $x_i\in U$.
We also define a scaled reward function:
\[
    r(s,\vec{a})=\frac{\rho(s,\vec{a})-\min_{s',\vec{a}'}\rho(s',\vec{a}')}{\sum_{s'',\vec{a}''}[\rho(s'',\vec{a}'')-\min_{s',\vec{a}'}\rho(s',\vec{a}')]}
\]
such that $r(s,\vec{a})$ forms a distribution over all pairs of $s$ and $\vec{a}$, i.e., $\forall s,\vec{a}.r(s,\vec{a})\geq 0$ and $\sum_{s,\vec{a}}r(s,\vec{a})=1$.
We will use the following variables:
\begin{itemize}
    \item $x_s^t\in S$: the state at stage $t$,
    \item $x_a^{i,t}\in A_i$: the action taken by Agent $i$ at stage $t$,
    \item $x_o^{i,t}\in O_i$: the observation received by Agent $i$ at stage $t$,
    \item $x_r^t\in S\times (A_1\times\ldots\times A_n)$: the reward earned at stage $t$,
    \item $x_T^t\in S$: transition distribution at stage $t$,
    \item $x_\Omega^t\in O_1\times\ldots\times O_n$: observation distribution at stage $t$,
    \item $x_p^t\in \mathbb{B}$: used to sum up rewards across stages.
\end{itemize}

\begin{figure*}[t]
    \centering
    \input{fig/derivation}
    \caption{The derivation of the induction case in the proof of Theorem~\ref{thm:reduction}.}
    \label{fig:derivation}
\end{figure*}

We represent elements in the sets $S$, $A_i$, and $O_i$ by integers, i.e., $S=\{0,1,\ldots,|S|-1\}$, etc., and use indices $s$, $a_i$, and $o_i$ to iterate through them, respectively.
On the other hand, a special treatment is required for variables $x_r^t$ and $x_\Omega^t$, as they range over Cartesian products of several sets.
We will give a unique number to an element in a product set as follows.
Consider $\vec{Q}=Q_1\times\ldots\times Q_n$, where each $Q_i$ is a finite set.
An element $\vec{q}=(q_1,\ldots,q_n)\in \vec{Q}$ is numbered by $N(q_1,\ldots,q_n)=\sum_{i=1}^n q_i(\prod_{j=1}^{i-1}|Q_j|)$.
In the following construction, variables $x_r^t$ and $x_\Omega^t$ will take values from the numbers given to the elements in $S\times\vec{A}$ and $\vec{O}$ by $N_r(s,\vec{a})$ and $N_\Omega(\vec{o})$, respectively.

We begin by constructing a DSSAT formula for a Dec-POMDP with $h=1$.
Under this setting, the derivation of the optimal joint policy is simplified to finding an action for each agent such that the expectation value of the reward function is maximized, i.e.,
\[
    \vec{a}^*=\argmax_{\vec{a}\in \vec{A}}\sum_{s\in S}\Delta_0(s)r(s,\vec{a})
\]
The DSSAT formula below encodes the above equation:
\[
    \invR x_s^0,\invR x_r^0,\exists x_a^{1,0}(D_{x_a^{1,0}}),\ldots,\exists x_a^{n,0}(D_{x_a^{n,0}}).\phi,
\]
where the distribution of $x_s^0$ follows $\Pr[x_s^0 \equiv s]=\Delta_0(s)$, the distribution of $x_r^0$ follows $\Pr[x_r^0 \equiv N_r(s,\vec{a})]=r(s,\vec{a})$, each $D_{x_a^{i,0}}=\emptyset$, and the matrix:
\[
    \phi=\bigwedge_{s\in S}\bigwedge_{\vec{a}\in\vec{A}}[x_s^0\equiv s\wedge\bigwedge_{i\in I} x_a^{i,0}\equiv a_i\rightarrow x_r^0\equiv N_r(s,\vec{a})].
\]
As the existentially quantified variables have no dependency on randomly quantified variable, the DSSAT formula is effectively an exist-random quantified SSAT formula.

For an arbitrary Dec-POMDP with $h>1$, we follow the two steps proposed in the previous work~\cite{SP19}, namely \textit{policy selection} and \textit{policy evaluation}, and adapt the policy selection step for the multi-agent setting in Dec-POMDP.

In the previous work~\cite{SP19}, an agent's policy selection is encoded by the following prefix (use Agent $i$ as an example):
\[
    \exists x_a^{i,0},\invR x_p^0,\invR x_o^{i,0},\ldots,\invR x_p^{h-2},\invR x_o^{i,h-2},\exists x_a^{i,h-1},\invR x_p^{h-1}.
\]
In the above quantification, variable $x_p^t$ is introduced to sum up rewards earned at different stages.
It takes values from $\mathbb{B}$, and follows a uniform distribution, i.e., $\Pr[x_p^t \equiv \top]=\Pr[x_p^t \equiv \bot]=0.5$.
When $x_p^t \equiv \bot$, the process is stopped and the reward at stage $t$ is earned; when $x_p^t \equiv \top$, the process is continued to stage $t+1$.
Note that variables $\{x_p^t\}$ are shared across all agents.
With the help of variable $x_p^t$, rewards earned at different stages are summed up with an equal weight $2^{-h}$.
Variable $x_o^{i,t}$ also follows a uniform distribution $\Pr[x_o^{i,t} \equiv o_i]=|O_i|^{-1}$, which scales the satisfying probability by $|O_i|^{-1}$ at each stage.
Therefore, we need to re-scale the satisfying probability accordingly in order to obtain the correct satisfying probability corresponding to the value of a joint policy.
The scaling factor will be derived in the proof of Theorem~\ref{thm:reduction}.

As the actions of the agents can only depend on their own observation history, for the selection of a joint policy it is not obvious how to combine the quantification, i.e., the selection of a policy, of each agent into a linearly ordered prefix required by SSAT, without suffering an exponential translation cost.
On the other hand, DSSAT allows to specify the dependency of an existentially quantified variable freely and is suitable to encode the selection of a joint policy.
In the prefix of the DSSAT formula, variable $x_a^{i,t}$ depends on $D_{x_a^{i,t}}=\{x_o^{i,0},\ldots,x_o^{i,t-1},x_p^0,\ldots,x_p^{t-1}\}$.

Next, the policy evaluation step is exactly the same as that in the previous work~\cite{SP19}.
The following quantification computes the value of a joint policy:
\begin{linenomath}
    \begin{align*}
        \invR x_s^t, \invR x_r^t, t=0,\ldots,h-1 \\
        \invR x_T^t, \invR x_\Omega^t, t=0,\ldots,h-2
    \end{align*}
\end{linenomath}
Variables $x_s^t$ follow a uniform distribution $\Pr[x_s^t \equiv s]=|S|^{-1}$ except for variable $x_s^0$, which follows the initial distribution specified by $\Pr[x_s^0 \equiv s]=\Delta_0(s)$; variables $x_r^t$ follow the distribution of the reward function $\Pr[x_r^t \equiv N_r(s,\vec{a})]=r(s,\vec{a})$; variables $x_T^t$ follow the state transition distribution $\Pr[x_{T_{s,\vec{a}}}^t \equiv s']=T(s,\vec{a},s')$; variables $x_\Omega^t$ follow the observation distribution $\Pr[x_{\Omega_{s',\vec{a}}}^t \equiv N_\Omega(\vec{o})]=\Omega(s',\vec{a},\vec{o})$.
Note that these variables encode the random mechanism of a Dec-POMDP and are hidden from agents.
That is, variables $x_a^{i,t}$ do not depend on the above variables.

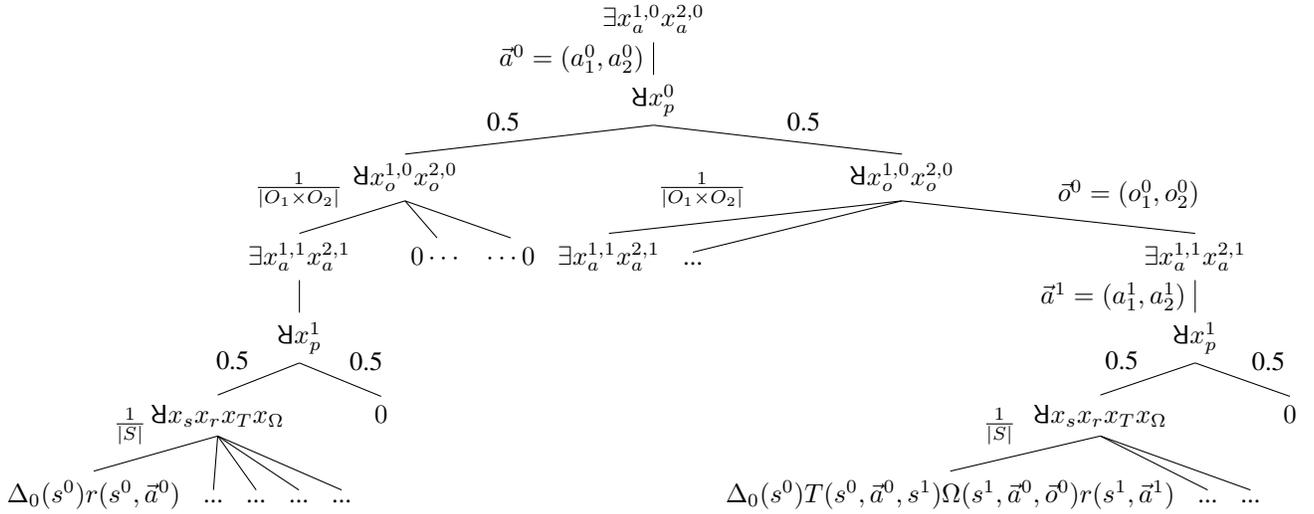
\begin{figure*}[t]
    \centering
    \begin{tikzpicture}
    \begin{scope}
    \Tree [.$\exists x_a^{1,0}x_a^{2,0}$ 
            \edge node[auto=right]{$\vec{a}^0=(a_1^0,a_2^0)$}; [.$\invR x_p^0$ 
              \edge node[auto=right]{0.5}; [.$\invR x_o^{1,0}x_o^{2,0}$ 
              \edge node[auto=right]{$\frac{1}{|O_1\times O_2|}$}; [.$\exists x_a^{1,1}x_a^{2,1}$ 
              [.$\invR x_p^1$ 
              \edge node[auto=right]{0.5}; [.$\invR x_sx_rx_Tx_\Omega$ 
              \edge node[auto=right]{$\frac{1}{|S|}$}; [.$\Delta_0(s^0)r(s^0,\vec{a}^0)$ ] 
              [.... ]
              [.... ]
              [.... ]
              [.... ] ] 
              \edge node[auto=left]{0.5}; [.$0$ ] ] ] 
              [.$0\cdots$ ] 
              [.$\cdots 0$ ] ] 
              \edge node[auto=left]{0.5}; [.$\invR x_o^{1,0}x_o^{2,0}$ 
              \edge node[auto=right]{$\frac{1}{|O_1\times O_2|}$}; [.$\exists x_a^{1,1}x_a^{2,1}$ ]
              [.... ]
              \edge node[auto=left]{$\vec{o}^0=(o_1^0,o_2^0)$}; [.$\exists x_a^{1,1}x_a^{2,1}$ 
              \edge node[auto=right]{$\vec{a}^1=(a_1^1,a_2^1)$}; [.$\invR x_p^1$ 
              \edge node[auto=right]{0.5}; [.$\invR x_sx_rx_Tx_\Omega$ 
              \edge node[auto=right]{$\frac{1}{|S|}$}; [.$\Delta_0(s^0)T(s^0,\vec{a}^0,s^1)\Omega(s^1,\vec{a}^0,\vec{o}^0)r(s^1,\vec{a}^1)$ ]
              [.... ] 
              [.... ] ]
              \edge node[auto=left]{0.5}; [.$0$ ] ] ] ] ] ]
    \end{scope}
    \end{tikzpicture}
    \caption{The decision tree of a Dec-POMDP example with two agents and $h=2$.}
    \label{fig:example}
    \vspace{-0.5cm}
\end{figure*}

The formulas to encode $\mathcal{M}$ are listed in Figure~\ref{fig:formula}.
Formula~\eqref{eq:xp_next} encodes that when $x_p^t \equiv \bot$, i.e., the process is stopped, the observation $x_o^{i,t}$ and next state $x_s^{t+1}$ are set to a preserved value $0$, and $x_p^{t+1} \equiv \bot$.
Formula~\eqref{eq:xp_stop} ensures the process is stopped at the last stage.
Formula~\eqref{eq:xr_0} ensures the reward at the first stage is earned when the process is stopped, i.e., $x_p^0 \equiv \bot$.
Formula~\eqref{eq:xr_t} requires the reward at stage $t>0$ is earned when $x_p^{t-1} \equiv \top$ and $x_p^t \equiv \bot$.
Formula~\eqref{eq:x_T} encodes the transition distribution from state $s$ to state $s'$ given actions $\vec{a}$ are taken.
Formula~\eqref{eq:x_omega} encodes the observation distribution to receive observation $\vec{o}$ under the situation that state $s'$ is reached after actions $\vec{a}$ are taken.

\begin{theorem}\label{thm:reduction}
    The above reduction maps a Dec-POMDP $\mathcal{M}$ to a DSSAT formula $\Phi$, such that a joint policy $\vec{\pi}$ exists for $\mathcal{M}$ if and only if a set of Skolem functions $\mathcal{F}$ exists for $\Phi$, with $V(\vec{\pi})=\Pr[\Phi|_{\mathcal{F}}]$.
\end{theorem}
\begin{proof}
    Given an arbitrary Dec-POMDP $\mathcal{M}$, a proof using mathematical induction over its planning horizon $h$ is as follows.

    For the base case $h=1$, to prove the ``only if'' direction, consider a joint policy $\vec{\pi}$ for $\mathcal{M}$ which specifies $\vec{a}=(a_1,\ldots,a_n)$ where agent $i$ will take action $a_i$. For this joint policy, the value is computed as $V(\vec{\pi})=\sum_{s \in S}\Delta_0(s)r(s,\vec{a})$. Based on $\vec{\pi}$, we construct a set of Skolem functions $\mathcal{F}$ where $x_a^{i,0}=a_i$ for each $i\in I$. To compute $\Pr[\Phi|_{\mathcal{F}}]$, we cofactor the matrix with $\mathcal{F}$ and arrive at the following CNF formula:
    \[
        \bigwedge_{s\in S}[x_s^0\neq s \vee x_r^0 \equiv N_r(s,\vec{a})],
    \]
    and the satisfying probability of $\Phi$ with respect to $\mathcal{F}$ is
    \begin{linenomath}
        \begin{align*}
            \Pr[\Phi|_{\mathcal{F}}] & =\sum_{s\in S}\Pr[x_s^0 \equiv s]\Pr[x_r^0 \equiv N_r(s,\vec{a})] \\
                                     & =\sum_{s\in S}\Delta_0(s)r(s,\vec{a})=V(\vec{\pi})
        \end{align*}
    \end{linenomath}

    Note in the above argument, only equalities are involved, and hence can be reversed to prove the ``if'' direction.

    For the induction case $h>1$, first assume that the statement holds. For a planning horizon of $h+1$, consider a joint policy $\vec{\pi}_{h+1}$ with value $V(\vec{\pi}_{h+1})$. Note that as a joint policy is a mapping from observation histories to actions, we can build a corresponding set of Skolem functions $\mathcal{F}_{h+1}$ to simulate joint policy $\vec{\pi}_{h+1}$ for the DSSAT formula. The derivation of satisfying probability with respect to $\mathcal{F}_{h+1}$ is shown in Figure~\ref{fig:derivation}. Note that to obtain the correct value of the joint policy, we need to re-scale the satisfying probability by a scaling factor $\kappa_{h+1}=2^{h+1}(|\vec{O}||S|)^{h}$.

    Once again, as only equalities are involved in the derivation in Figure~\ref{fig:derivation}, the ``if'' direction is also proved.

    As $\Pr[\Phi|_{\mathcal{F}_{h+1}}]=V(\vec{\pi}_{h+1})$, the theorem is proved according to the principle of mathematical induction.
\end{proof}

\subsection{Discussion}
Below we count the numbers of variables and clauses in the resulting DSSAT formula with respect to the input size of the given Dec-POMDP.
For one stage, there are $3+2(|I|+|S||\vec{A}|)$ variables, and therefore in total the number of variables is $O(h(|I|+|S||A|))$ asymptotically.
On the other hand, the number of clauses per stage is $2+|I|+|S||\vec{A}|+|S|^2|\vec{A}|+|S||\vec{A}||\vec{O}|$, and hence the total number of clauses is $O(h(|I|+|S||\vec{A}|(|S|+|\vec{O}|))$.
Overall, we show that the proposed reduction is polynomial-time with respect to the input size of the Dec-POMDP.

Below we demonstrate the reduction with an example.

\begin{example}
    Consider a Dec-POMDP with two agents and planning horizon $h=2$.
    Given a joint policy $(\pi_1,\pi_2)$ for Agent $1$ and Agent $2$, let the actions taken at $t=0$ be $\vec{a}^0=(a_1^0,a_2^0)$ and the actions taken at $t=1$ under certain observations $\vec{o}^0=(o_1^0,o_2^0)$ be $\vec{a}^1=(a_1^1,a_2^1)$.
    The value of this joint policy is computed by Eq.~\eqref{eq:bellman} as
    \begin{linenomath}
        \begin{align*}
            V(\pi) & =\sum_{s^0\in S}\Delta_0(s^0)[r(s^0,\vec{a}^0)                                                                  \\
                   & +\sum_{\vec{o}^0\in\vec{O}}\sum_{s^1\in S}T(s^0,\vec{a}^0,s^1)\Omega(s^1,\vec{a}^0,\vec{o}^0)r(s^1,\vec{a}^1)].
        \end{align*}
    \end{linenomath}

    The decision tree of the converted DSSAT formula is shown in Figure~\ref{fig:example}.
    At $t=0$, after taking actions $\vec{a}^0$, variable $x_p^0$ splits into two cases: when $x_p^0\equiv \bot$ (left branch), the expected reward $\Delta_0(s^0)r(s^0,\vec{a}^0)$ will be earned for $t=0$; on the other hand, when $x_p^0\equiv \top$ (right branch), observation $\vec{o}^0$ is received, based on which the agents will select their actions $\vec{a}^1$ at $t=1$.
    Again, variable $x_p^1$ will split into two cases, but this time $x_p^1$ is forced to be $\bot$ as it is the last stage.
    The expected reward $\Delta_0(s^0)T(s^0,\vec{a}^0,s^1)\Omega(s^1,\vec{a}^0,\vec{o}^0)r(s^1,\vec{a}^1)$ will be earned under the branch of $x_p^1\equiv\bot$ for $t=1$.
    Note that the randomized quantifiers over variables $x_p^t$, $x_s^t$, and $x_o^t$ will scale the satisfying probability by the factors labelled on the edges, respectively.
    Therefore, we have to re-scale the satisfying probability by $2^2|S||O_1\times O_2|$, which is predicted by the scaling factor $\kappa_{h}=2^h(|\vec{O}||S|)^{h-1}$ calculated in the proof of Theorem~\ref{thm:reduction}.
\end{example}
\section{Conclusions and Future Work}\label{sec:conclusions}
In this paper, we extended DQBF to its stochastic variant DSSAT and proved its NEXPTIME-completeness.
Compared to the PSPACE-complete SSAT, DSSAT is more powerful to succinctly model NEXPTIME-complete decision problems with uncertainty.
The new formalism can be useful in applications such as artificial intelligence and system design.
Specifically, we demonstrated the DSSAT formulation of the analysis to probabilistic/approximate partial design, and gave a polynomial-time reduction from the NEXPTIME-complete Dec-POMDP to DSSAT.
We envisage the potential broad applications of DSSAT and plan solver development for future work.
We note that recent developments of \textit{clausal abstraction} for QBF~\cite{JanotaM15,RabeT15} and DQBF~\cite{Tentrup19} might provide a promising framework for DSSAT solving.
Clausal abstraction has been lifted to SSAT~\cite{ChenHJ21}, and we are investigating its feasibility for DSSAT.

\section*{Acknowledgments}
The authors are grateful to Christoph Scholl, Ralf Wimmer, and Bernd Becker for valuable discussions motivating this work.
This work was supported in part by the Ministry of Science and Technology of Taiwan under Grant No.~108-2221-E-002-144-MY3, 108-2218-E-002-073, and 109-2224-E-002-008.
JHJ was supported in part by the Alexander von Humboldt Foundation during this work.

\bibliography{references}
\end{document}